\documentclass[11pt]{article}

\usepackage[T1]{fontenc}
\usepackage{textcomp}
\usepackage{palatino}
\usepackage{mathpazo}
\usepackage{stmaryrd}

\usepackage{hyperref}
\hypersetup{pdfpagemode=UseNone}

\usepackage{amsfonts}
\usepackage{amssymb}
\usepackage{amsmath}
\usepackage{latexsym}
\usepackage{amsthm}
\usepackage{eepic}
\usepackage{color}
\usepackage{enumitem}

\usepackage{mathdots}

\newtheorem{theorem}{Theorem}
\newtheorem{proposition}[theorem]{Proposition}

\newtheorem{definition}{Definition}

\newtheorem{corollary}[theorem]{Corollary}

\newtheorem{theorem*}{Theorem}
\newtheorem{corollary*}{Corollary}

\def\Tr{\textnormal{Tr}}

\def\Supp{\textnormal{Supp }}

\def\<{\langle}
\def\>{\rangle}

\def\Seq{\textnormal{Seq}}

\newcommand{\commentout}[1]{}

\numberwithin{theorem}{section}
\numberwithin{equation}{section}

\newcommand{\talkingPoint}[1]{}
\title{Universal security for randomness expansion from the spot-checking
protocol}

\author{%
   Carl~A.~Miller and Yaoyun~Shi  \\
 \\
  Department of Electrical Engineering and Computer Science\\
  University of Michigan, Ann Arbor, MI 48109, USA\\
  \texttt{{carlmi,shiyy@umich.edu}}
}

\date{}

\begin{document}
\maketitle
\thispagestyle{empty}

\begin{abstract}\noindent
Colbeck (\textit{Thesis}, 2006) proposed using Bell inequality violations to generate
certified random numbers.  While full quantum-security proofs have been given,
it remains a major open problem to identify the broadest class of Bell inequalities
and lowest performance requirements to achieve such security.  In this paper,
working within the broad class of \textit{spot-checking protocols}, we prove exactly
which Bell inequality violations can be used to achieve full security.
Our result greatly improves the known noise tolerance for secure randomness expansion: for the commonly used CHSH game, full security was only known with a noise tolerance of 1.5\%, and we improve this to 10.3\%.  We also generalize our results
beyond Bell inequalities and give the first security proof for randomness
expansion based on Kochen-Specker inequalities.
The central technical contribution of the paper is a new uncertainty principle for the Schatten
norm, which is based on the uniform convexity inequality of Ball, Carlen, and Lieb (\textit{Inventiones mathematicae}, 115:463--482, 1994).
\end{abstract}

\clearpage
\setcounter{page}{1}

\section{Introduction}\label{sec:motivation}
Randomness is indispensable for modern day information processing.
Secure generation of randomness is at the very foundation of modern cryptography:
a message is secretive precisely when
it is random to the adversary. But the generation of true randomness 
is challenging both theoretically and practically. A fundamental difficulty is the fact
that there is no complete test of randomness --- all randomness tests can be fooled by
a deterministic generator. 
And indeed, current solutions have problems: for example, Heninger {\em et al.}~\cite{Heninger:2012} and Lenstra {\em et al.}~\cite{Lenstra+}
have found independently that a significant percentage of cryptographic keys can be broken due to
the lack of entropy. One of NIST's previous standards for pseudorandom number generation is widely believed to be backdoored~\cite{NYT:Snowden}.

More recently, RNGs based on quantum measurements have emerged in the market (e.g., IDQuantique's Quantis).
While a (close to) perfect implementation of certain measurements can theoretically
guarantee randomness, current technology is still far from reaching that precision.
An additional concern is, even if in the future the implementation technology is satisfactory,
could there be backdoors in the generator inserted by a malicious party? It is difficult for the user,
as a classical being, to directly verify the inner-working of the quantum device. 

In his Ph.D. thesis~\cite{col:2006}, Colbeck formulated the problem of untrusted-device randomness expansion, and proposed a protocol based on non-local games. In Colbeck's protocol, a nonlocal game is played repeatedly (using
an initial random seed) on
a multi-part quantum device.  The proposed basis for security is that if the device exhibits a superclassical
average score, then it must be
exhibiting quantum, and therefore random, behavior.
If a ``success'' event occurs in the protocol --- if the average score is above
a certain acceptance threshhold --- then the outputs produced by the device
are assumed to be partially random, and a quantum-proof
randomness extractor (e.g., \cite{De:2012}) is applied to produce a shorter string which represents the final
output of the protocol.  

The goal of the aforementioned protocol is that the final output be a uniformly
random string, thus achieving  \textit{randomness
expansion} --- starting with initial random seed and obtaining
a larger random output.  Subsequent authors observed that one can
minimize the size of the seed by skewing the input distribution used for the nonlocal game (e.g., by 
giving the device a fixed input string on most rounds).

Proving security for such a protocol requires showing
that, when the protocol succeeds, there is a uniform lower bound $H_{min}^\delta ( X \mid
E ) \geq C$ on the smooth min-entropy of $X$ (the raw outputs of the device)
conditioned on $E$ (quantum information that may be possessed by 
an adversary).  
This bound must be proved to hold for all quantum devices
compatible with the nonlocal game.
Once this is established, it follows that any
quantum-proof randomness extractor
can be used to produce roughly $C$ random bits \cite{R05}.  (See \cite{BertaFS, ChungLW} for recent
work on quantum-proof randomness extractors.)

Several authors proved security of Colbeck-type protocols against classical side information only~\cite{pam:2010, fehr13, pironio13, CoudronVY:2013}. In a groundbreaking work, Vazirani and Vidick~\cite{Vazirani:dice} proved
full quantum security. Their protocol provided
exponential randomness expansion, and was later extended to unbounded expansion \cite{CY:STOC}.  The quantum protocol of \cite{Vazirani:dice}
included an 
exact requirement on the performance of the device, and so it remained an
open question whether quantum security could be proved for a robust
(error-tolerant) protocol.

In~\cite{MS14v3}, the present authors proved quantum security with error-tolerance,
via a new set of techniques based on Renyi entropies.    We proved other new features (cryptographic
security, unit size quantum memory, nonzero bit-rate) and also  proved, with Chung and Wu, that error-tolerant unbounded randomness
expansion was possible (\cite{CSW14:short, MS14v3, Gross}).
Our protocol 
used a class of $n$-player binary XOR games, and the noise tolerance
was significant ($\geq 1.5\%$) but not proven to be optimal.

The motivating question for the current paper is this: what are the most general conditions under which randomness expansion
can be proved with full quantum security?
Answering this question is important for lowering the implementation
requirements on randomness expansion protocols, and thus
improving their practical value. 

\subsection{Central result}

\label{centralresultsubsec}

We wish to answer the following questions:
\begin{enumerate}
\item What devices can achieve secure randomness
expansion?

\item What games can be used to achieve secure randomness expansion?

\item What is the largest noise tolerance for each game?
\end{enumerate}
In this paper we will limit our focus in the following way: we will deal only with protocols in which the Bell inequality (or other device 
test) is the same on each round, and is applied using independently generated inputs.
(This is a natural assumption, although we note that the literature
includes protocols that do not satisfy this assumption, e.g., \cite{Vazirani:dice}.  
Answering questions 1-3 for such multi-stage protocols is another interesting
avenue for research.)

Given that a protocol is using the same Bell inequality on each round, 
the only way to achieve superlinear randomness expansion is to adjust
the input distribution for the Bell inequality so that a certain fixed input
$\overline{a}$ occurs with probability $(1-q)$, where $q$ is a parameter
that we make small when the number of rounds ($N$) is large.  (Without this requirement,
the amount of seed needed to generate the random inputs would have to be at least linearly proportional to the
number of extractable bits in the output.)  This motivates a protocol
which we import from \cite{CoudronVY:2013}, referred
to here as the \textit{spot-checking protocol}.  
A general
version of the spot-checking protocol is shown in Figure~\ref{roverlineprot}.
(In \cite{CoudronVY:2013}, this protocol was showed to be
secure against classical adversaries.)

Let $E$ be an
external quantum system which may be entangled with the device $D$.
Let $A, X, T$ be classical registers denoting, respectively, the collected
inputs of $D$, the collected outputs of $D$, and the bits $(t_1, \ldots, t_N)$
across all $N$ rounds.  Let $\Gamma = \Gamma_{AXTE}$ denote the joint state
of these registers taken together, and let $\Gamma^s \leq \Gamma$ denote
the subnormalized operator corresponding to the ``success'' event.  If
the normalization of $\Gamma^s$ satisfies
\begin{eqnarray}
H_{min} \left( X \mid ATE \right) \geq y,
\end{eqnarray}
we say that Protocol $R_{gen}$ has produced $y$ extractable bits.
If $\Gamma_s$ is within trace-distance $\delta$ of an operator
$\Gamma'$ satisfying the same condition, or is within trace-distance
$\delta$ of the zero state, then
we say that Protocol $R_{gen}$
has produced $y$ extactable bits with soundness error $\delta$.

For any nonlocal game $G$, let $W_G$ be the supremum
of all expected scores that can be achieved at $G$ by compatible quantum devices,
and let $W_{G , \overline{a}}$ denote the same supremum taken
just over devices that give deterministic outputs on input $\overline{a}$.
Our central result is summed up by the following theorem.

\begin{theorem}
\label{mainthmsimple}
For any game $G$, there are functions $\pi \colon [0, W_G] \to \mathbb{R}$ 
and $\Delta \colon (0, 1]^2 \to \mathbb{R}$ such that the following
hold.

\begin{enumerate}
\item For any $b \in (0, 1]$, Protocol $R_{gen}$ produces at least
\begin{eqnarray}
N \left[ \pi ( \chi ) - \Delta ( b, q) \right]
\end{eqnarray} 
extractable bits with soundness error $3 \cdot 2^{-bqN}$.

\item The function $\pi$ is nonzero on the interval $( W_{G, \overline{a}} , W_G ]$.

\item The function $\Delta$ tends to $0$
as $(b,q) \to (0, 0 )$.
\end{enumerate}
\end{theorem}

Briefly stated, the theorem asserts that Protocol $R_{gen}$ always achieves
secure expansion provided that $\chi > W_{G, \overline{a}}$.  
The function $\Delta$ is an error term which vanishes when the
test probability $q$ and the soundness error term $b$ are sufficiently small.
Crucially, this function
depends only on $G$ and not on any other parameters in
Protocol $R_{gen}$.  The bound is therefore device-independent.

We note that the noise tolerance in this result is optimal: 
if the noise threshhold $\chi$ were less than
$W_{G, \overline{a}}$, then Protocol $R_{gen}$
clearly could not produce superlinear expansion,
since the device can always behave deterministically
on generation rounds.  We therefore have complete answers to questions 1--3 above.
The upper limit for the noise tolerance for $G$ is $W_{G, \overline{a}}$;  
the games that can be used in the spot-checking protocol
are precisely those that have $W_{G, \overline{a}} < W_G$.  And the devices
that can be used in the spot-checking protocol are precisely those
which exceed $W_{G, \overline{a}}$ for some $G$.

In a previous draft of this paper, we informally claimed
that \textit{all} superclassical devices can be used for exponential randomness
expansion, but this was an overstatement of our results.
Actually, there are superclassical devices $D$ that do not exceed
the threshhold $W_{G, \overline{a}}$ for any $G$ (see Appendix~\ref{magicapp}).
Such devices cannot be effectively used in the spot-checking protocol,
but could potentially be used in other randomness expansion protocols, and studying their use remains
an open problem.

Following a method used in \cite{CoudronVY:2013},
if we let $q = (\log^2 N )/ N$ in Protocol $R_{gen}$, 
then it will require only $O_G ( \log^3 N )$ bits of initial seed to approximate
the input distribution $TA$, and $O_G ( \log^2 N )$ bits of initial seed
to perform randomness extraction on $X$ once the protocol concludes
\cite{De:2012}.  Since the number of final random bits
is $\Theta ( N )$, exponential randomness expansion is achieved.

\subsection{Beyond nonlocal games}

Our proof of Theorem~\ref{mainthmsimple} builds on methods
from our previous paper \cite{MS14v3}.  
In that paper we focused on untrusted devices $D$ that make
measurements $\{ P_\mathbf{a}^\mathbf{x} \}$ of the form
\begin{eqnarray}
P_\mathbf{a}^\mathbf{x} & = & P_{1, a_1}^{x_1} \otimes
\cdots \otimes P_{n, a_n}^{x_n}.
\end{eqnarray}
(where $\mathbf{a}$ denotes an input sequence and $\mathbf{x}$ denotes
an output sequence).
This tensor product form reflects the fact that $D$ has $n$ spatially separated
components.  In the present paper we considered whether this assumption
can be replaced with a different assumption about the measurements of $D$.

Recent papers \cite{HorodeckiH:2010, AbbottC:2012, Um:contextual, DengZ:2013}
have analyzed randomness expansion protocols based on \textit{contextuality},
rather than spatial separation.  We were therefore motivated to generalize
our proof to include contextual randomness expansion as well.  We do this
by defining a \textbf{contextual device} to be a device which accepts input
sequences $\mathbf{a} = (a_1, \ldots, a_m )$ (not necessarily of uniform length)
and returns output sequences $\mathbf{x} = (x_1, \ldots, x_m )$, with the only restriction
that the measurements performed by the device have the form
\begin{eqnarray}
P_\mathbf{a}^\mathbf{x} & = & P_{1, a_1}^{x_1} P_{2, a_2}^{x_2}
\cdot \cdots \cdot P_{m, a_m}^{x_m},
\end{eqnarray}
and the individual measurements $\{ \{ P_{1, a_i}^{x_i} \}_{x_i} \}_{i=1}^m$ are required
to be simultaneously diagonalizable.   (See Definition~\ref{contextualdef}.)
We prove security for such devices in parallel to spatially separated devices, thus showing that contextuality alone is sufficient
to prove quantum-secure randomness expansion.  This is also a new achievement.

\subsection{Rate curves}

An interesting unsolved problem left by our work is maximizing the
rate curves $\pi$ in Theorem~\ref{mainthmsimple}.  The rate curve proved
in Theorem~\ref{fixedinputthm} is the following:
\begin{eqnarray}
\label{quotedrc}
\pi_G ( x ) & = & \left\{ \begin{array}{cl}
\frac{ 2 (\log e ) ( x - W_{G, \overline{a}} )^2 }{r - 1 } & \textnormal{ if }
x > W_{G, \overline{a}} \\ \\
 0 & \textnormal{ otherwise,} 
\end{array} \right.
\end{eqnarray}
where $r$ is the size of the total output alphabet for the device $D$.
This is clearly not optimal: for example, we proved in \cite{MS14v3} that the rate curve for the GHZ game
is at least
\begin{eqnarray}
x \mapsto 1 - 2 H ( (1-x) / 0.11)
\end{eqnarray}
for $x > 0.89$, where $H$ denotes the Shannon entropy function,
and this exceeds
(\ref{quotedrc}) for values of $x$ near $1$.  Improving these curves
is vital for maximizing the performance of quantum random number
generation.  This problem is distilled in Section~\ref{rcsec}
(in terms of the function $X \mapsto \Tr [X^{1+\epsilon}]$)
and can be studied as a separate problem.

\begin{figure}[h]
\begin{center}
\setlength{\fboxsep}{10pt}
\setlength{\fboxrule}{1pt}
\fbox{\parbox{\linewidth}{

\textit{Arguments:}

\begin{enumerate}[itemsep=-3pt,topsep=3pt]
\item[$G$]: A  game with a distinguished input $\overline{a}$.
\item[$D$]: A quantum device compatible with $G$.
\item[$N$]: A positive integer (the \textbf{output length}). 
\item[$q$]: A real number from the interval $( 0, 1 )$.  (The
\textbf{test probability}.)
\item[$\chi$]: A real number from the interval $( 0, 1)$.  (The
\textbf{score threshhold}.)
\end{enumerate} 

\textbf{Protocol $R_{gen}$:}

\begin{enumerate}[itemsep=-3pt,topsep=3pt]
\item Let $c$ denote a real variable which
we initially set to $0$.

\item Choose a bit $t \in \{ 0, 1 \}$ according
to the distribution $(1-q,q )$.  

\item If $t = 1$ (``game round''), then the game
$G$ is played with $D$ and the output is recorded.
The score achieved is added to the variable $c$.

\item If $t = 0$ (``generation round'') then $\overline{a}$ is given to 
$D$ and the output is recorded.

\item Steps 2--4 are repeated $(N-1)$ more times.

\item If $c  <  \chi q N$, 
then the protocol \textbf{aborts}.  Otherwise, it \textbf{succeeds}.
\end{enumerate}
}}

\caption{Protocol $R_{gen}$ (modified from \cite{CoudronVY:2013})}
\label{roverlineprot}
\end{center}
\end{figure}

\newpage

\tableofcontents

\section{An overview of proof techniques}

Our proof begins in the same setting as our previous paper \cite{MS14v3}.
Let $\rho \colon
E \to E$ be a density matrix.  Then, one way to measure the randomness of $\rho$
is via the von Neumann entropy:
\begin{eqnarray}
H ( \rho ) & = & - \Tr [ \rho \log \rho ].
\end{eqnarray}
This quantity measures, asymptotically, the number of random bits that can
be extracted from $\rho^{\otimes m}$ as $m \to \infty$~\cite{R05,TomamichelCR:2009,Tomamichel:thesis}.  This is not
directly useful in the setting of untrusted-device cryptography, since there
is typically no reliable way to produce independent copies of a single state.
But consider instead the quantity
\begin{eqnarray}
H_{1+\epsilon} ( \rho ) & = & - \frac{1}{\epsilon} \log \Tr [ \rho^{1+\epsilon} ],
\end{eqnarray}
the Renyi entropy of the eigenvalues of $\rho$ (which tends to $H ( \rho )$ as $\epsilon \to 0$, although not uniformly).  The smooth
min-entropies of $\rho$ satisfy
\begin{eqnarray}
H_{min}^\epsilon ( \rho ) & \geq & H_{1+\epsilon} ( \rho ) - \frac{\log ( 1 / \delta ) }{\epsilon}.
\end{eqnarray}
Thus, function $H_{1+\epsilon} ( \rho )$ can be used to provide lower bounds
on the number of extractable bits in the cryptographic setting --- provided that
$\epsilon$ is not too small.

Thus, our goal becomes to prove in general that a device $D$ that scores high
at a nonlocal game $G$ must exhibit a randomness in its output, as measured
by the ``$(1+\epsilon)$-randomness''
function $X \mapsto - \frac{1}{\epsilon} \log \Tr ( X^{1+\epsilon} )$.  However there is
an apparent difficulty in proving such a direct relationship: the score of a device
is defined in terms of the function $X \mapsto \Tr ( X )$, which
cannot be uniformly determined from $X \mapsto
\Tr ( X^{1+\epsilon} )$.  We therefore introduce
in this paper the notion of the ``$(1+\epsilon)$-score''
of the device $D$ at the game $G$ (see Definition~\ref{schscoredef}), which is
a quantity
defined in terms of $\Tr ( X^{1+\epsilon} )$ which tends to the ordinary expected
score as $\epsilon \to 0$.  

We prove in section~\ref{rcsec} that the function $\pi_G$ in 
\ref{quotedrc} provides a universal lower bound on the $(1+\epsilon)$-randomness
of a device $D$, on input $\overline{a}$, in terms of the $(1+\epsilon)$-score
at the game $G$.  (Thus, $\pi_G$ is a ``rate curve'' for $G$.)  For this
we use the following uncertainty principle, which is the main new technical contribution
in this paper (see section~\ref{randschattensec}):
Let $\left\| \cdot \right\|_{1+\epsilon}$ denote the
$(1+\epsilon)$-Schatten matrix norm, which is defined by
\begin{eqnarray}
\left\| Z \right\|_{1+\epsilon}
= [\Tr ( (Z^* Z)^{\frac{1+\epsilon}{2}} )]^{\frac{1}{1+\epsilon}}.
\end{eqnarray}
\begin{proposition}
\label{binuncertaintypropquote}
For any finite dimensional Hilbert space $V$, any positive semidefinite
operator $\tau \colon V \to V$ satisfying $\left\| \tau \right\|_{1+\epsilon} = 1$,
and any binary projective measurement $\{ R_0, R_1 \}$ on $V$, the following
holds.  Let
\begin{eqnarray}
\tau' = R_0 \tau R_0 + R_1 \tau R_1.
\end{eqnarray}
Then
\begin{eqnarray}
\left\| \tau' \right\|_{1+\epsilon} & \leq &
1 - (\epsilon/2) \left\| \tau - \tau' \right\|^2_{1+\epsilon}.
\end{eqnarray}
\end{proposition}
This principle can be understood as follows: if the measurement
$\{ R_0, R_1 \}$ significantly disturbs the state $\tau$
(as measured by the Schatten norm) then
the amount of randomness in the post-measurement
state is significantly more than that of $\tau$.
The proof of this result is a based on a known result \cite{BallC:1994}
on the uniform convexity of $\left\|  \cdot \right\|_{1+\epsilon}$.
We also prove a similar result for non-binary measurements.

The proof of the rate curve then uses contrapositive reasoning:
if input $\overline{a}$ to device $D$ does not produce
much randomness, then the initial state of $D$ must have
already been close (under $\left\| \cdot \right\|_{1+\epsilon}$) to that of a device that gave predictable outputs on input $\overline{a}$, and therefore its $(1+\epsilon)$-score at $G$ cannot be much larger than $W_{G, \overline{a}}$.
With the rate curve established, we prove the final security
result (Corollary~\ref{finalcor}) in section~\ref{protsec} via
a simplified version of the methods used in \cite{MS14v3}.

Our 
proof in this paper is largely self-contained --- the two primary outside results
that we rely on are the uniform convexity result, and a theorem on the relationship
between Renyi entropies and smooth min-entropy (Theorem~\ref{translatetominthm}).

\section{Definitions and notation}

\label{defsec}

\subsection{Device models}
\label{devmodelsubsec}

For our purposes, a \textit{quantum device component} is an
object containing a quantum system $Q$ which
receives a classical input ($a \in \mathcal{A}$) at each use, performs a corresponding
operation on an internal quantum system, and then produces a
classical output ($x \in \mathcal{X}$).  This process is
repeated a finite number of times.
It is assumed that the device can maintain
(quantum) memory in between rounds and may be entangled with
other devices and external quantum systems.  We can assume without loss
of generality that the quantum operations performed on $Q$
by such a component at each round are the same
(since a device that uses different operations depending
on previous inputs and outputs can always be simulated by one that
maintains a transcript in memory).  Also, using the Stinespring representation theorem, we can assume that the quantum operations
performed consist of an $\mathcal{X}$-valued projective measurement
on $Q$ (which may depend on the input $a$) followed by a unitary automorphism of $Q$.

\begin{definition}
\label{gendevdef}
A quantum device with input alphabet $\mathcal{A}$
and output alphabet $\mathcal{X}$ consists of the following data.
\begin{enumerate}
\item A finite dimensional Hilbert space $Q$ and a density
operator $\phi \colon Q \to Q$.

\item For each $a \in \mathcal{A}$, a projective measurement
$\{ P_a^x \}_{x \in \mathcal{X}}$ on $Q$ and a unitary
automorphism $U_a \colon Q \to Q$.
\end{enumerate}
\end{definition}

\begin{definition}
\label{componentdef}
A quantum device with $r$ components is a quantum device
satisfying the following additional conditions.
\begin{enumerate}
\item[3.] The space $Q$ has the form $Q = Q_1 \otimes \cdots \otimes
Q_r$ and the alphabets have the form $\mathcal{A} =
\mathcal{A}_1 \times \cdots \times \mathcal{A}_r$ and 
$\mathcal{X} =
\mathcal{X}_1 \times \cdots \times \mathcal{X}_r$.

\item[4.] The projective measurements $\{P_a^x \}_x$ have the form
\begin{eqnarray}
P_a^x & = & P_{a_1, 1}^x \otimes \cdots \otimes P_{a_r , r}^x
\end{eqnarray}
where $\{ P_{a_i,i}^x \}_{x \in \mathcal{X}}$ is a projective
measurement on $Q_i$ for each $a_i ,i$.
\end{enumerate}
\end{definition}

These devices operate by first performing the projective measurement
$\{ P_a^x \}$ on $Q$ and outputting the result, and then applying
the unitary automorphism $U_a$.
We note that there is no restriction placed on the
unitary automorphisms $U_a$ in Definition~\ref{componentdef}, which means that this model allows the components of the device $D$ to communicate with one another in between uses.

We also make the following definition.  For any finite set $S$, let $\Seq (S )$
denote the set of non-repeating sequences of elements from $S$.

\begin{definition}
\label{contextualdef}
A contextual device is a quantum device (Definition~\ref{gendevdef}) satisfying
the following additional conditions.
\begin{enumerate}
\item[3.] There are finite sets $\mathcal{B}, \mathcal{Y}$ such that
$\mathcal{A} \subseteq \Seq ( \mathcal{B} )$ and $\mathcal{X} =
\Seq ( \mathcal{Y} ) $.

\item[4.] There are projective measurements $\{ P_b^y \}_y$ 
for each $b \in \mathcal{B}$ such that
for any $a = (b_1, \ldots, b_m ) \in \mathcal{A}$, the measurements $\{ \{ P_b^y \}_y
\mid b \in \{ b_1, \ldots, b_m \} \}$ are simultaneously diagonalizable,
and 
\begin{eqnarray}
P_a^x & = & P_{b_1}^{y_1} P_{b_2}^{y_2} \cdots P_{b_m}^{y_m}
\end{eqnarray}
for any $m$-length sequence $x = (y_1, \ldots, y_m ) \in \mathcal{X}$.
\end{enumerate}
\end{definition}

The following definition is convenient for some of the proofs that will follow.
\begin{definition}
An \textbf{abstract} quantum device is defined as
in Definition~\ref{gendevdef}, but with the assumption
that $\phi$ is a density operator replaced by the weaker
assumption that $\phi$ is a nonzero positive semidefinite operator.
\end{definition}

We will use the following notation: if $D$ is a device with initial state
$\phi \colon Q \to Q$, and $X$ is a positive semidefinite operator
on $Q$, then
\begin{eqnarray}
\phi_X & := & \sqrt{X} \phi \sqrt{X}.
\end{eqnarray}
and
\begin{eqnarray}
\rho_X & := & \left( \sqrt{\phi} X \sqrt{\phi} \right)^\top.
\end{eqnarray}
The intuitions for these operators are as follows: if $\{ M , 
\mathbb{I} - M \}$ is a binary measurement on $Q$, then $\phi_M$
represents the post-measurement state of $Q$ for outcome $M$, and $\rho_M$ represents the corresponding post-measurement
state of a purifying system $\overline{Q}$ for $Q$.  

The operators $\rho_X$ are crucial objects of study for establishing
full security of a protocol involving the device $D$.  The purifying
system $\overline{Q}$ represents the maximal amount
of quantum information that an adversary could posses about the
device $D$.  We will refer to the operators $\rho_X$ as the \textbf{adversary states} of $D$, and to the operators
$\phi_X$ as the \textbf{device states} of $D$.  Note that $\rho_X$ has the same
singular values as $\phi_X$.

If an indexed sequence $(z_1, z_2, \ldots )$ is given, we use the boldface
variable $\mathbf{z}$ to denote the entire sequence.
For any sequences $(a_1, \ldots, a_n ) \in \mathcal{A}^n$
and $(x_1, \ldots, x_n ) \in \mathcal{X}^n$, let
\begin{eqnarray}
\phi_\mathbf{a}^\mathbf{x} & = & M_n M_{n-1} \cdots M_1 \phi M^*_1 \cdots M^*_{n-1} M^*_n, \\ \nonumber \\
\rho_\mathbf{a}^\mathbf{x} & = & \left( \sqrt{\phi} M_1^* M_2^* \cdots M_n^* M_n
\cdots M_2 M_1 \sqrt{\phi} \right)^\top,
\end{eqnarray}
where
\begin{eqnarray}
M_j & = & U_{a_j} P_{a_j}^{x_j}.
\end{eqnarray}
These represent the device states and adversary states occuring
for the input and output sequences $\mathbf{a}, \mathbf{x}$.

\subsection{Games}

We will state a general definition of a game.  In this paper, we will frequently
make use of a fixed input ($\overline{a}$) for the game, and so for 
convenience we make that choice of input part of the definition.

\begin{definition}
\label{gamedef}
A \textbf{game} $G$ consists of the following
data.
\begin{enumerate}
\item A finite set $\mathcal{A}$ (the \textbf{input alphabet}) with
a distinguished element $\overline{a} \in \mathcal{A}$
and a probability distribution $p \colon \mathcal{A} \to [0, 1]$.

\item A finite set $\mathcal{X}$ (the \textbf{output alphabet}).

\item A \textbf{scoring function} $H \colon \mathcal{A} \times \mathcal{X}
\to [0, 1]$.
\end{enumerate}
\end{definition}

The game operates as follows: an input $a$
is chosen according to the probability distribution $p$, and it is given to a device,
which produces an output $x$.  The function $H$
is applied to $(a, x )$ to obtain the score.  

The following is a notational convenience: if $\mathbf{a} \in \mathcal{A}^n$
and $\mathbf{x} \in \mathcal{X}^n$, then let
\begin{eqnarray}
p ( \mathbf{a} ) & = & p ( a_1) p(a_2) \cdots p ( a_n ) \\
H ( \mathbf{a} , \mathbf{x} ) & = & H ( a_1, x_1 ) + \ldots + H ( a_n , x_n ).
\end{eqnarray}

\begin{definition}
An \textbf{unbounded} game is defined as in 
Definition~\ref{gamedef}, except that the function $H$ maps to
$[0, \infty)$.
\end{definition}

\begin{definition}
\label{ngdef}
A \textbf{nonlocal game} $G$ with $s$ players
is a game in which the input and output alphabets
have the form $\mathcal{A} = \mathcal{A}_1 \times \cdots
\times \mathcal{A}_s$ and $\mathcal{X} = \mathcal{X}_1
\times \cdots \times \mathcal{X}_s$.
\end{definition}

\begin{definition}
A \textbf{contextual game} $G$ is
a game in which the input and output alphabets
satisfy $\mathcal{A} \subseteq \Seq ( \mathcal{B} )$
and $\mathcal{X} = \Seq ( \mathcal{Y} )$ for some
finite sets $\mathcal{B}, \mathcal{Y}$.
\end{definition}

A quantum device is \textbf{compatible} with a game if it
has the same descriptor (nonlocal or contextual) 
and the input
alphabet and output alphabet match those of the game (including
their decompositions into Cartesian products or
sequence-sets, as appropriate).

We can combine the above definitions as follows: if $G$ is a game
and $D$ is a compatible device, then the expected score of $D$ for $G$ (on the first use) is given by
\begin{eqnarray}
\sum_{a \in \mathcal{A}, x \in \mathcal{X}} 
p(a) H (a, x ) \Tr ( P_a^x \phi ).
\end{eqnarray}

\begin{definition}
Let $G$ be a  game.  Then, the \textbf{quantum value}
of $G$, denoted $W_G$, is the supremum of the expected scores for $G$ taken
over all devices compatible with $G$. 
\end{definition}

In order for a device to be useful for our purposes, it must generate
a score at some nonlocal game which guarantees quantum
behavior on a \textit{fixed} input string (as in \cite{CoudronVY:2013}).  This
motivates the following definitions.

\begin{definition}
Let $D$ be an abstract quantum device and let $\overline{a} \in \mathcal{A}$ be
a fixed input letter.  Then, we say that $D$ is \textbf{deterministic}
on input $\overline{a}$ if
\begin{eqnarray}
\phi &  = & P_{\overline{a}}^{\overline{x}} \phi P_{\overline{a}}^{\overline{x}} 
\end{eqnarray}
for some output letter $\overline{x}$.  We say that $D$ is
\textbf{classically predictable} on input $\overline{a}$ if 
\begin{eqnarray}
\phi & = & \sum_{x \in \mathcal{X}} P_{\overline{a}}^x \phi P_{\overline{a}}^x.
\end{eqnarray}
\end{definition}

\begin{definition}
If $G$ is a  game, then let $W_{G, \overline{a}}$ denote
the supremum of the expected scores for $G$ over all
compatible devices $D$ that have classically
predictable outputs on input $\overline{a}$.
\end{definition}

\subsection{Device-independent vanishing error functions}

It is necessary to be cautious about our use of asymptotic notation because
ultimately we will want to assert that the bounds we prove on the
rate of our protocol (including second-order terms) are truly
device-independent.
We adopt the following conventions for asymptotic notation: If we say,

\vskip0.1in

``For all $x,y \in (0, 1]$, $F ( x, y ) \leq O ( y )$.''

\vskip0.1in

Then the coefficient and the range in the big-$O$ expression
must be independent of $x$.  (The sentence above asserts that there
is a function $G(y )$ satisfying $\lim_{y \to 0 } G ( y )/y < \infty$ such that $F(x, y ) \leq G ( y )$.)
On the other hand, if we say

\vskip0.1in

``Let $x \in (0, 1]$ be a real number.  Then, for all $y \in (0, 1]$, $F ( x, y ) \leq O ( y )$.''

\vskip0.1in

Then it is understood that the coefficient and the range in the big-$O$
expression may depend on $x$.  (The sentences above assert that $\lim_{y \to 0}
F ( x, y)/y < \infty$ for all $x \in (0, 1 ]$.)

\vskip0.1in

Equivalently, we may write

\vskip0.1in

``For all $x, y \in (0, 1]$, $F(x, y ) \leq O_x ( y )$.''

\vskip0.1in

Here the subscripted $x$ indicates that dependence on $x$ is allowed. 

\section{The functions $\left\| \cdot \right\|_{1+\epsilon}$ and $\left<  \cdot \right>_{1+\epsilon}$}

\label{funcsec}

We continue with preliminaries in this section.
For any linear operator $Z$ and any $\ell \geq 1$, the $\ell$-Schatten norm is given by
\begin{eqnarray}
\left\| Z \right\|_\ell & = & \Tr \left[  (Z^* Z)^{\frac{\ell}{2}} \right]^{\frac{1}{\ell}}.
\end{eqnarray}
(If $W$ is positive semidefinite, then the $\ell$-Schatten norm
can be written more simply as
$\left< W \right>_\ell = \Tr \left[ W^\ell \right]^{\frac{1}{\ell}}$.)
For convenience, we will also define the notation $\left< \cdot \right>_\ell$ to mean the
following
simpler expression:
\begin{eqnarray}
\left< Z \right>_\ell & = & \Tr \left[  (Z^* Z)^{\frac{\ell}{2}} \right].
\end{eqnarray}
We have $\left< Z \right>_\ell = \left\| Z \right\|_\ell^\ell$.
We will often be concerned with the functions $\left< W \right>_{1+\epsilon}$
and $\left\| W \right\|_{1+\epsilon}$ for small $\epsilon$.

A discussion of some relevant properties of the Schatten norm
can be found in \cite{Rastegin:2012}.  For our purposes, we will need
the following properties of $\left< \cdot \right>_{1+\epsilon}$ and
$\left\| \cdot \right\|_{1+\epsilon}$.
The functions are almost linear in the following
sense: for positive semidefinite operators $X,Y$,
\begin{eqnarray}
\label{eprel1}
& (1 - O ( \epsilon )) \left( \left\| X \right\|_{1+\epsilon} +
\left\| Y \right\|_{1+\epsilon} \right) \leq \left\| X + Y \right\|_{1+\epsilon}
\leq \left\| X \right\|_{1+\epsilon} +
\left\| Y \right\|_{1+\epsilon} \\ \nonumber \\
\label{eprel2}
& \left< X \right>_{1+\epsilon} + \left< Y  \right>_{1+\epsilon}
\leq \left< X + Y \right>_{1+\epsilon} \leq
(1+ O ( \epsilon) ) \left( \left< X \right>_{1+\epsilon}
+ \left< Y  \right>_{1+\epsilon} \right).
\end{eqnarray}
Also, the righthand inequalities in (\ref{eprel1}) and (\ref{eprel2}) hold also
when $X$ and $Y$ are replaced by arbitrary linear operators (not necessarily
positive semidefinite).

When $A$ is a positive semidefinite operator on a space $V = V_1 \oplus \ldots \oplus V_m$, and $A' = \sum_k P_k A P_k$, where $P_k$ denotes projection on $V_k$, then
\begin{eqnarray}
\label{measuresch1}
(1 - O_m (\epsilon ) ) \left\| A \right\|_{1+\epsilon} \leq 
\left\| A' \right\|_{1+\epsilon} \leq \left\| A \right\|_{1+\epsilon} \\
\label{measuresch2}
(1 - O_m (\epsilon ) ) \left< A \right>_{1+\epsilon} \leq 
\left< A' \right>_{1+\epsilon} \leq \left< A \right>_{1+\epsilon},
\end{eqnarray}
and the righthand inequalities both hold also when $A$ is replaced by an arbitrary
linear operator.

Unless otherwise specified, the domain of the variable $\epsilon$
will always be $(0, 1]$.

\subsection{Relationship to extractable bits}

\label{smminsubsec}

Our motivation for studying the function $\left< \cdot \right>_{1+\epsilon}$
is its relationship to quantum smooth min-entropy.
The smooth min-entropy of a classical-quantum $CQ$ state measures (asymptotically)
the number of bits that can be extracted from $C$ in the presence
of an adversary who possesses $Q$ \cite{R05}.  
\begin{definition}
Let $CQ$ be a classical-quantum system whose state is $\Gamma_{CQ}$.  Then,
the min-entropy of $C$ conditioned on $Q$ is
\begin{eqnarray}
H_{min} \left( C \mid Q \right) & = & 
\max_{\mathbb{I}_C \otimes \sigma \geq \alpha } 
\left[ - \log \Tr (\sigma ) \right],
\end{eqnarray}
where $\sigma$ varies over all positive semidefinite operators on $Q$ that satisfy
the given inequality. For any
$\delta > 0$, the min-entropy of $C$ conditioned on $Q$ with smoothing parameter
$\delta$ is
\begin{eqnarray}
H^\delta_{min} \left( C \mid Q \right) & = & 
\max_{\left\| \Gamma' - \Gamma_{CQ} \right\|_1 \leq \delta } 
H_{min} \left( C \mid Q \right)_{\Gamma'},
\end{eqnarray}
where $\Gamma'$ varies over all classical-quantum positive semidefinite
operators on $CQ$ that satisfy the given inequality.
\end{definition}

In this paper we will implicitly use the
quantum Renyi divergence functions developed by \cite{JaksicOPP:2011, MullerDSFT:2013,
WildeWY:2013}, and surveyed recently by \cite{Tomamichel:survey}. In order to conserve space, we will not introduce a full
formalism for these functions here, but will just note the following intuition:
Let $CQ$ be a classical quantum system whose state
is given by $\sum_c \left| c \right> \left< c \right| \otimes \alpha_c $, and let
$\alpha = \sum_c \alpha_c$.  Then, the expression
\begin{eqnarray}
- \frac{1}{\epsilon} \log \left( \sum_c \left< \alpha_c \right>_{1+\epsilon} \right)
\end{eqnarray}
can be thought of as an absolute measure of the amount of randomness
contained in $CQ$, while the related expression
\begin{eqnarray}
- \frac{1}{\epsilon} \log \left( \sum_c \left< \alpha^{\frac{- \epsilon}{2+2\epsilon} }\alpha_c \alpha^{\frac{- \epsilon}{2+2\epsilon} }
\right>_{1+\epsilon} \right)
\end{eqnarray}
can be thought of as a measure of the amount of randomness in $C$ conditioned
on $Q$.\footnote{The latter quantity is, in formal terms, the negation of the Renyi divergence of the state of $CQ$ relative to the state $\mathbb{I}_C \otimes \alpha$.}  This intuition
is supported by the following theorem about smooth min-entropy.
\begin{theorem}
\label{translatetominthm}
Let $\Lambda$ be a subnormalized operator on a bipartite system $CQ$ of the form $\Lambda = \sum_c \alpha_c \otimes \left| c \right> \left< c \right|$, and 
let $\sigma$ be a density matrix on $Q$.  Let 
\begin{eqnarray}
\label{expabove}
K &= & - \frac{1}{\epsilon} \log \left( \sum_c \left< \sigma^{\frac{- \epsilon}{2+2\epsilon} }\alpha_c \sigma^{\frac{- \epsilon}{2+2\epsilon} }
\right>_{1+\epsilon} \right)
\end{eqnarray}
Then, for any $\delta > 0$,
\begin{eqnarray}
\label{thethm}
H_{min}^\delta \left( C \mid Q \right)_\Lambda \geq K - \frac{1 + 2 \log ( 1 / \delta )}{\epsilon}.
\end{eqnarray}
\end{theorem}

\begin{proof}
Corollary D.8 in \cite{MS14v3} proves the above (building on
\cite{TomamichelCR:2009, DupuisFS:2013}) with the extra
assumption that $\Tr ( \Lambda ) = 1$.  The case $\Tr ( \Lambda ) \leq 1$
follows by rescaling and using the fact that $1+\epsilon \leq 2$.
\end{proof}

(For a more detailed discussion of
the relationship between smooth min-entropy
and Renyi divergence, see section~$6.4.1$ of
\cite{Tomamichel:survey}, which uses a different definition of 
$H^\delta_{min} \left( \cdot \mid \cdot \right)$.)

\section{Randomness versus state disturbance}

\label{randschattensec}

Our central goal is to prove the randomness of
certain classical variables in the presence of quantum side information,
using the Schatten norm as a metric.   This section provides
inductive steps for such proofs of randomness.

The basis for all of the results in this section is a known
result on the \textit{uniform convexity} of the Schatten norm, Theorem 1 of \cite{BallC:1994}.  We state the following
proposition, which is a special case of uniform convexity.

\begin{proposition}
\label{curvatureprop}
For any $\epsilon \in (0,1]$, and any linear operators
$W$ and $Z$ such that $\left\| W \right\|_{1+\epsilon} = \left\| Z
\right\|_{1+\epsilon} = 1$,
\begin{eqnarray}
\left\| \frac{W + Z}{2} \right\|_{1 + \epsilon} \leq
1 - \frac{\epsilon}{8} \left\| W - Z \right\|_{1+\epsilon}^2
. \qed
\end{eqnarray}
\end{proposition}

\begin{proof}
Substituting $X = (W+Z)/2$, $Y = (W-Z)/2$, and $p = 1 + \epsilon$
into Theorem~1
of \cite{BallC:1994}, we have
\begin{eqnarray}
1 & \geq & \left\| \frac{W+Z}{2} \right\|_{1+\epsilon}^2 +
\epsilon \left\| \frac{W - Z}{2} \right\|_{1+\epsilon}^2
\end{eqnarray}
which implies $\left\| (W+Z)/2 \right\|_{1+\epsilon}^2 \leq 1 - 
(\epsilon/4) \left\| W - Z \right\|_{1+\epsilon}^2$. 
The proof is completed by the fact that $\sqrt{1 - x} \leq 1 - (x/2)$
for any $x \in [0,1]$.  
\end{proof}

The next proposition compares the amount of randomness obtained
from a measurement
to the degree of disturbance
in the state that is caused by the measurement.

\begin{proposition}
\label{binuncertaintyprop}
For any finite dimensional Hilbert space $V$, any positive semidefinite
operator $\tau \colon V \to V$ satisfying $\left\| \tau \right\|_{1+\epsilon} = 1$,
and any binary projective measurement $\{ R_0, R_1 \}$ on $V$, the following
holds.  Let
\begin{eqnarray}
\tau' = R_0 \tau R_0 + R_1 \tau R_1.
\end{eqnarray}
Then
\begin{eqnarray}
\left\| \tau' \right\|_{1+\epsilon} & \leq &
1 - (\epsilon/2) \left\| \tau - \tau' \right\|^2_{1+\epsilon}.
\end{eqnarray}
\end{proposition}

\begin{proof}
Choose a basis $\{ e_1, \ldots, e_n \}$ for $V$
such that $R_0$ is the projector into the space spanned
by $e_1, \ldots, e_m$, and write $\tau$ in $(m, n-m)$-block form:
\begin{eqnarray}
\label{writingtau}
\tau = \left[ \begin{array}{c|c} T_{00} & T_{01} \\
\hline
T_{10} & T_{11} \end{array} \right].
\end{eqnarray}
Applying Proposition~\ref{curvatureprop} with 
$W = \tau$ and
\begin{eqnarray}
Z = \left[ \begin{array}{c|c} T_{00} & - T_{01} \\
\hline
- T_{10} & T_{11} \end{array} \right],
\end{eqnarray}
we obtain
\begin{eqnarray}
\left\| \tau' \right\|_{1+\epsilon} 
& = & \left\| (W + Z)/2 \right\|_{1+\epsilon} \\
& \leq & 1 - \frac{\epsilon}{8} \left\| W - Z \right\|_{1+\epsilon}^2 \\ 
& = & 1 - \frac{\epsilon}{8} \left\| 
2 ( \tau - \tau') \right\|_{1+\epsilon}^2 \\
& = & 1 - \frac{\epsilon}{2} \left\| 
\tau - \tau' \right\|_{1+\epsilon}^2,
\end{eqnarray}
as desired.
\end{proof}

\begin{proposition}
\label{ternaryprop}
For any finite dimensional Hilbert space $V$, any positive semidefinite
operator $\tau \colon V \to V$ satisfying $\left\| \tau \right\|_{1+\epsilon} = 1$,
and any binary projective measurement $\{ R_0, R_1 , \ldots, R_n \}$ on $V$, the following
holds.  Let
$\tau' = \sum_i P_i \tau P_i$.  Then
\begin{eqnarray}
\left\|  \tau' \right\|_{1+\epsilon} & \leq &
1 - \frac{\epsilon}{2n} \left\| \tau - \tau' \right\|^2_{1+\epsilon} + O_n ( \epsilon^2).
\end{eqnarray}
\end{proposition}

\begin{proof}
For any $i \in \{ 1, \ldots, n \}$, let
\begin{eqnarray}
\tau_i & = & \left( P_0 \tau P_0 + \cdots + P_{i-1} \tau P_{i-1} \right) + (P_i + \cdots + P_n ) \tau (P_i + \cdots + P_n ),
\end{eqnarray}
and let $\tau_0 = \tau$.
Note that applying Proposition~\ref{binuncertaintyprop} to
the measurement $\{ R_n , \mathbb{I} - R_n \}$ and
the state $\tau_{n-1}/\left\| \tau_{n-1} \right\|_{1+\epsilon}$ yields
\begin{eqnarray}
\left\| \tau_n \right\|_{1+\epsilon} & \leq & 
\left[ 1 - \frac{\epsilon}{2}  \left( \frac{ \left\| \tau_n
- \tau_{n-1} \right\|_{1+\epsilon} }{ \left\| \tau_{n-1}
\right\|_{1+\epsilon}} \right)^2 \right] \left\| \tau_{n-1} \right\|_{1+\epsilon} \\
& \leq & 
\left[ 1 - \frac{\epsilon}{2}  \left\| \tau_n
- \tau_{n-1} \right\|_{1+\epsilon}^2 \right]
\left\| \tau_{n-1} \right\|_{1+\epsilon}.
\end{eqnarray}
By an inductive argument, we then have
\begin{eqnarray}
\left\| \tau_n \right\|_{1+\epsilon} & \leq & 
\prod_{i=1}^n \left( 1 - \frac{\epsilon}{2}  \left\| \tau_i
- \tau_{i-1} \right\|_{1+\epsilon}^2 \right) \cdot 1\\
& \leq & 1 - \frac{\epsilon}{2} \sum_{i=1}^n 
\left\| \tau_i
- \tau_{i-1} \right\|_{1+\epsilon}^2 + O_n ( \epsilon^2) \\
& \leq & 1 - \frac{\epsilon}{2n} \left( \sum_{i=1}^n 
\left\| \tau_i
- \tau_{i-1} \right\|_{1+\epsilon} \right)^2
+ O_n (\epsilon^2) \\
& \leq & 1 - \frac{\epsilon}{2n} \left\| \tau_n -  \tau_0 \right\|_{1+\epsilon}^2
+ O_n ( \epsilon^2),
\end{eqnarray}
The operator $\tau_n$ is equal to $\tau'$, and this completes the proof.
\end{proof}

\section{Rate curves}

\label{rcsec}

Our next goal is prove inequalities which relate
the randomness generated by a device to its performance
at a given game.  First we state the following
alternative version of Proposition~\ref{ternaryprop},
which follows easily from the fact that
$\left< X \right>_{1+\epsilon} = \left\| X \right\|_{1+\epsilon}^{1+\epsilon}$.

\begin{proposition}
\label{ternaryprop2}
For any finite dimensional Hilbert space $V$, any positive semidefinite
operator $\tau \colon V \to V$ satisfying $\left< \tau \right>_{1+\epsilon} = 1$,
and any binary projective measurement $\{ R_0, R_1 , \ldots, R_n \}$ on $V$, the following
holds.  Let
$\tau' = \sum_i P_i \tau P_i$.  Then
\begin{eqnarray}
\left<  \tau' \right>_{1+\epsilon} & \leq &
1 - \frac{\epsilon}{2n} \left< \tau - \tau' \right>^2_{1+\epsilon} + O_n ( \epsilon^2). \qed
\end{eqnarray}
\end{proposition}

\subsection{The $(1+\epsilon)$-score of a game}

\label{parameterssubsec}

\begin{definition}
\label{gameopdef}
Let $G = (p, H )$ be a  game (with alphabets $\mathcal{A}, \mathcal{X}$), and let $D$ be a compatible abstract
device.  Then, the \textbf{game operator} for $D$ determined by $G$ is
\begin{eqnarray}
K := \sum_{\substack{a \in \mathcal{A} \\ x \in \mathcal{X} }} p( a ) H ( a, x ) P_a^x.
\end{eqnarray}
where $\{ \{ P_a^x \}_x \}_a$ denote the measurements performed by $D$.  Let the expressions
 $\rho_G, \phi_G$ denote the operators
\begin{eqnarray}
\rho_G := \left( \sqrt{\rho} K \sqrt{\rho} \right)^\top \hskip1in \phi_G := \sqrt{K} \rho \sqrt{K}.
\end{eqnarray}
\end{definition}

We note that since all of the terms $H(a,x)$ are assumed to be less than or
equal to $1$, it follows that the game operator always satisfies $K \leq \mathbb{I}$.  

Some intuition for the expression $\rho_G$ is as follows: Let $Q'$ be a quantum system of the same dimension
as the system $Q$ inside $D$, and suppose that $QQ'$ is in a pure state which purifies $Q$.  
Suppose that the game $G$ is played, a score $h \in [0, 1 ]$ is obtained, and then
a classical random variable $X \in \{ P, F \}$ is set to be equal to $P$ with probability
$h$, and $F$ with probability $(1-h)$.  Then, $\rho_G$ is isomorphic to
the subnormalized state of $Q'$ corresponding to the event $X = P$.

\begin{definition}
\label{schscoredef}
Let $G = (p, H)$ be a  game with input alphabet $\mathcal{A}$ and
output alphabet $\mathcal{X}$, and let $D$ be a compatible abstract device.  Let $\epsilon \in [0, 1 ]$.
Then, the \textbf{$(1+\epsilon)$-score} of $D$ (for $G$) is given by
\begin{eqnarray}
\label{ratecurveexp}
W_G^\epsilon ( D ) & := & \frac{  \left< \phi_G \right>_{1+\epsilon}}{\left< \phi \right>_{1+\epsilon}}.
\end{eqnarray}
\end{definition}

Let $W_G ( D ) = W_G^0 ( D )$.  Note that if $D$ is an ordinary quantum device this is simply the expected score of the device $D$ at the
game $D$.

\begin{proposition}
\label{upperlimitprop}
Let $G$ be a  game with distinguished input $\overline{a}$. Then for any
compatible abstract device $D$ whose output on input $\overline{a}$ is classically
predictable, we
must have
\begin{eqnarray}
W_G^\epsilon ( D ) \leq W_{G, \overline{a} } + O (\epsilon ).
\end{eqnarray}
\end{proposition}

\begin{proof}
First consider the case where $D$ behaves deterministically 
on input $\overline{a}$.  Let
\begin{eqnarray}
D = ( Q, \phi, \{ P_a^x \} , \{ U_a \}).
\end{eqnarray}
We have
$\Supp \phi \subseteq \Supp P_{\overline{a}}^{\overline{x}}$ for some $\overline{x}$.  Let
$K$ denote the game operator for $D,G$.

The inequality
\begin{eqnarray}
P_{\overline{a}}^{\overline{x}} K P_{\overline{a}}^{\overline{x}}
\leq \left( W_{G, \overline{a}} \right) P_{\overline{a}}^{\overline{x}}
\end{eqnarray}
must hold, since if it did not, we could find a unit vector $v \in
\Supp P_{\overline{a}}^{\overline{x}}$ such that $v K v^* >
W_{G, \overline{a}}$, which would mean the device
\begin{eqnarray}
D' = ( Q, v v^*, \{ P_a^x \} , \{ U_a \})
\end{eqnarray}
breaks the score threshhold $W_{G, \overline{a}}$ and
gives deterministic output on $\overline{a}$, a contradiction.
Therefore,
\begin{eqnarray}
\left< \phi_G \right>_{1+\epsilon} & = & \left< \sqrt{K} \phi \sqrt{K}
\right>_{1+\epsilon} \\
& = & \left< \sqrt{\phi} K \sqrt{\phi} \right>_{1+\epsilon} \\
& = & \left< \sqrt{\phi} P_{\overline{a}}^{\overline{x}} K 
P_{\overline{a}}^{\overline{x}} \sqrt{\phi} \right>_{1+\epsilon} \\
& \leq & W_{G, \overline{a}} \left< \sqrt{\phi} P_{\overline{a}}^{\overline{x}}
\sqrt{\phi}  \right>_{1+\epsilon} \\
& = & W_{G, \overline{a}} \left< \phi \right>_{1+\epsilon},
\end{eqnarray}
as desired.

Now consider the case where $D$ is classically predictable
on input $\overline{a}$.  In this case, the state of $D$ is a
convex combination of states that would give deterministic
output on input $\overline{a}$, and thus the approximate
linearity of $\left\< \cdot \right>_{1+\epsilon}$ (see (\ref{eprel2}))
yields the desired result.
\end{proof}

When $D$ is a quantum device and $\overline{a}$ is an input letter, we measure
the amount of randomness produced by $D$ on input $\overline{a}$ by comparing
the values of the function $\left< \cdot \right>_{1+\epsilon}$ applied
to the premeasurement and postmeasurement states of $D$.

\begin{definition}
\label{schrandomnessdef}
Let $D$ be a quantum device and let $\overline{a}$ be an input letter
for $D$.  
Then, the \textbf{$(1+\epsilon)$-randomness} of $D$ for input
$\overline{a}$ is the following quantity:
\begin{eqnarray}
\label{ratecurveexp}
- \frac{1}{\epsilon} \log \left( \frac{\sum_{x \in \mathcal{X}} \left< \phi_{\overline{a}}^x \right>_{1+\epsilon}}{
\left< \phi \right>_{1+\epsilon} } \right).
\end{eqnarray}
Let $G$ be a  game with which $D$ is compatible.  Then,
the $(1+\epsilon)$-randomness of $D$ for the game $G$ is the following
quantity:
\begin{eqnarray}
\label{ratecurveexp}
- \frac{1}{\epsilon} \log \left( \frac{\sum_{a \in \mathcal{A},x \in \mathcal{X}} p ( a ) \left< \phi_a^x \right>_{1+\epsilon}}{
\left< \phi \right>_{1+\epsilon} } \right),
\end{eqnarray}
where $p$ denotes the input distribution for $G$.
\end{definition}

\subsection{Universal rate curves for fixed input}

\begin{definition}
Let $G$ be a  game, and let
$\pi \colon [0, W_G ] \to \mathbb{R}_{\geq 0}$ be a nondecreasing 
convex function which is 
differentiable on $(0, W_G )$.  Then, $\pi$ is a \textbf{rate curve for $(G, \overline{a})$} if for
all compatible abstract devices $D$,
\begin{eqnarray}
- \frac{1}{\epsilon} \log \left( \frac{  \sum_x \left< \phi_{\overline{a}}^x \right>_{1+\epsilon}  }{ \left< \phi \right>_{1+\epsilon } } \right) 
 & \geq & \pi ( W^{\epsilon}_G ( D ) ) - O ( \epsilon ).
\end{eqnarray}
\end{definition}

\begin{definition}
Let $G$ be a  game, and let
$\pi \colon [0, W_G ] \to \mathbb{R}_{\geq 0}$ be a nondecreasing 
convex function which is 
differentiable on $(0, W_G )$.  Then, $\pi$ is a \textbf{rate curve for $G$} if for
all compatible abstract devices $D$,
\begin{eqnarray}
- \frac{1}{\epsilon} \log \left( \frac{ \sum_{a,x} p ( a ) \left< \phi_a^x \right>_{1+\epsilon}  }{ \left< \phi \right>_{1+\epsilon } } \right) 
 & \geq & \pi ( W^{\epsilon}_G ( D ) ) - O ( \epsilon ).
\end{eqnarray}
\end{definition}

\begin{theorem}
\label{fixedinputthm}
Let $G$ be a   game with output alphabet
size $r \geq 2$, let $\overline{a}$ be the distinguished input letter for $G$, and let $w
= W_{G, \overline{a}}$.  Then, the following function is a rate curve for $(G, \overline{a})$.
\begin{eqnarray}
\pi ( x ) & = & \left\{ \begin{array}{cl}
\frac{ 2 (\log e ) ( x - w )^2 }{r - 1 } & \textnormal{ if }
x > w \\ \\
 0 & \textnormal{ otherwise.} 
\end{array} \right.
\end{eqnarray}
\end{theorem}

\begin{proof}
It suffices to give a proof
for abstract devices $D$ satisfying $\left< \phi \right>_{1+\epsilon} = 1$
and $W_G^\epsilon ( D ) > w$, so we will assume those two conditions
in what follows.  Following previous notation, let $\phi' = \sum_x \phi_{\overline{a}}^x$,
and let $K$ denote the game operator.
Let $D'$
denote the device $D$ with the operator $\phi$ replaced by $\phi'$.
We will prove the desired rate curve by comparing
the distance between $\phi$ and $\phi'$ the difference between
their $(1+\epsilon)$-scores.  
 
Note that for any Hermitian operator $X$, the value of the operator
\begin{eqnarray}
\left[ \begin{array}{c|c} \sqrt{K} X \sqrt{K} & \sqrt{K} X \sqrt{\mathbb{I} - K} \\ \hline
\sqrt{\mathbb{I} - K} X \sqrt{ K} & \sqrt{\mathbb{I} - K} X \sqrt{\mathbb{I} - K} \end{array}
\right]
\end{eqnarray}
under $\left< \cdot \right>_{1+\epsilon}$ is the same as that of $X$, because the
two operators are unitarily equivalent.  Thus, using the discussion
of inequalities (\ref{measuresch1}--\ref{measuresch2}), we have
\begin{eqnarray}
\left< X \right>_{1+\epsilon} & \geq & 
\left< \sqrt{K} X \sqrt{K} \right>_{1+\epsilon} + 
\left< \sqrt{ \mathbb{I} - K} X \sqrt{\mathbb{I} - K} \right>_{1+\epsilon}.
\end{eqnarray}
Applying this for $X = \phi - \phi'$, followed by the general inequality
\begin{eqnarray}
\left< Y - Z \right>_{1+\epsilon} \geq \left( 1+ O (\epsilon ) \right) \left| \left< Y \right>_{1+
\epsilon} - \left< Z \right>_{1+\epsilon} \right|,
\end{eqnarray}
 yields
\begin{eqnarray}
\left< \phi - \phi' \right>_{1+\epsilon}  & \geq & 
\left< \sqrt{K} \phi \sqrt{K} \right>_{1+\epsilon} - \left< \sqrt{K} \phi' \sqrt{K} \right>_{1+\epsilon} \\
& & + 
\left< \sqrt{\mathbb{I} - K} \phi' \sqrt{
\mathbb{I} - K} \right>_{1+\epsilon} - \left< \sqrt{\mathbb{I} -  K} \phi \sqrt{
\mathbb{I} - K} \right>_{1+\epsilon} - O ( \epsilon )  \\
& \geq & 
W_G^\epsilon ( D) - W^\epsilon_G ( D') + [1 - W^\epsilon_G ( D' ) ] - [1 - W^\epsilon_G ( D ) ]  - O ( \epsilon) \\
\label{keyineqagain}
& = & 2 [ W^\epsilon_G ( D ) -  W^\epsilon_G ( D' )] - O ( \epsilon ).
\end{eqnarray}

Note that by Proposition~\ref{upperlimitprop}, we have
\begin{eqnarray}
W_G^\epsilon ( D' ) & \leq &  (w + O ( \epsilon ) )
\left< \phi' \right>_{1+\epsilon} \\
& \leq & (w + O ( \epsilon ) )
\left< \phi \right>_{1+\epsilon} \\
& = & (w + O ( \epsilon ) ),
\end{eqnarray}
and thus (\ref{keyineqagain}) implies
\begin{eqnarray}
\left< \phi - \phi' \right>_{1+\epsilon}  & \geq &
2 [ W^\epsilon_G ( D ) -  w ] - O ( \epsilon ).
\end{eqnarray}
Substituting this relationship into Proposition~\ref{ternaryprop} yields
\begin{eqnarray}
 \left< \phi' \right>_{1+\epsilon}  & \leq & 1 - 
\frac{2\epsilon }{r-1} \left( W_G ( D )  -  w  \right)^2
+ O ( \epsilon^2),
\end{eqnarray}
which implies the desired rate curve.
\end{proof}

\section{A general security proof}

\label{protsec}

\label{secpfapp}

In this section we provide a general proof of security for Protocol $R_{gen}$
in Figure~\ref{roverlineprot}.  The method of proof
is a generalization of that of our previous paper
on this topic \cite{MS14v3}.

Note that for any device $D$, each device-state $\phi_\mathbf{a}^{\mathbf{x}}$
has the same singular values as the corresponding adversary-state
$\rho_\mathbf{a}^{\mathbf{x}}$, and so any calculation involving the singular values
of the first can be rewritten in terms of the second, and vice versa.  Since this section
is concerned with proving security against an adversary, we will focus on expressions
involving the adversary states.

\subsection{The weighted $(1+\epsilon)$-randomness function}

If $G$ is an unbounded  game, and $D$ is an abstract device compatible
with $G$, then let $R_G^\epsilon (D )$ denote the $(1+\epsilon)$-randomness
of $D$ for $G$ (Definition~\ref{schrandomnessdef}):
\begin{eqnarray}
R_G^\epsilon (D ) & = & - \frac{1}{\epsilon} \log \left( \frac{ 
\sum_{a,x} p ( a ) \left< \phi_a^x \right>_{1+\epsilon} }{\left< \phi
\right>_{1+\epsilon} } \right).
\end{eqnarray}
Equivalently,
\begin{eqnarray}
R_G^\epsilon (D ) & = & - \frac{1}{\epsilon} \log \left( \frac{ 
\sum_{a,x} p ( a ) \left< \rho_a^x \right>_{1+\epsilon} }{\left< \rho
\right>_{1+\epsilon} } \right).
\end{eqnarray}
Also if $a$ denotes an input for $G$, let $R_{a}^\epsilon ( D )$ denote the $(1+\epsilon)$-randomness of $D$ on input $\overline{a}$:
\begin{eqnarray}
R_{a}^\epsilon (D ) & = & - \frac{1}{\epsilon} \log \left( \frac{ 
\sum_{x}  \left< \rho_{a}^x \right>_{1+\epsilon} }{\left< \rho
\right>_{1+\epsilon} } \right).
\end{eqnarray}

For $s \in \mathbb{R}$, let $R_G^{\epsilon, s } ( D )$ denote the expression
\begin{eqnarray}
R_G^{\epsilon, s} (D ) & = & - \frac{1}{\epsilon} \log \left( \frac{ 
\sum_{a,x} p ( a ) 2^{\epsilon s H ( a, x )} \left< \rho_a^x \right>_{1+\epsilon} }{\left< \rho \right>_{1+\epsilon} } \right),
\end{eqnarray}
which we will call the $(1+\epsilon)$-randomness of $D$ for $G$, weighted by $s$.
(This quantity is central to the inductive argument in our security proof.)

\begin{proposition}
\label{devrelnprop}
Let $G$ be a  game and let
$s \in \mathbb{R}$.  Then, for all compatible devices $D$,
\begin{eqnarray}
R_G^{\epsilon,s} ( D ) \geq 
R_G^{\epsilon} ( D ) - s W_G^{\epsilon} ( D ) - O ( \epsilon ).
\end{eqnarray}
\end{proposition}

\begin{proof}
For simplicity, we prove the result for abstract
devices $D$ satisfying $\left< \rho \right>_{1+\epsilon} = 1$.  (The general
case then follows easily.)
We have the following, using the fact that $2^t = 1 + (\ln 2 ) t + O ( t^2 )$.
\begin{eqnarray}
&& \sum_{a,x} p ( a ) 2^{\epsilon s H ( a, x )} \left< \rho_a^x \right>_{1+\epsilon}  \\
& = &
\sum_{a,x} p ( a )  \left< \rho_a^x \right>_{1+\epsilon} + 
\sum_{a,x} p ( a )  (2^{\epsilon s H ( a, x )} - 1 ) \left< \rho_a^x \right>_{1+\epsilon} \\
& \leq &
\sum_{a,x} p ( a )  \left< \rho_a^x \right>_{1+\epsilon} + 
\sum_{a,x} p ( a )  (\ln 2 ) \epsilon s H ( a, x) \left< \rho_a^x \right>_{1+\epsilon} 
+  O ( \epsilon^2) \\
& \leq &
\left[ 1 - \epsilon (\ln 2 ) R_G^\epsilon ( D ) \right] + \epsilon s (\ln 2 ) W_{G}^\epsilon ( D )
+ O ( \epsilon^2).
\end{eqnarray}
Applying the function $-\frac{1}{\epsilon} \log ( \cdot )$ to both sides
yields the desired result.
\end{proof}

It is desirable to have a lower bound on the weighted randomness quantity that
is device-independent.  For this, we will use the following principle.  Suppose that $R \colon [a,b] \to \mathbb{R}$ is a differentiable convex function, and we wish to compute
$\min_{x \in [a,b]} R ( x )$.  Then, there are three possibilities:
\begin{enumerate}
\item $R' ( t )  > 0$ for all $t$, in which case $\min_x R ( x) = R ( a )$.

\item $R' ( t ) < 0$ for all $t$, in which case $\min_x R ( x ) = R  ( b ) $.

\item There exists $t_0$ such that $R' ( t_0 ) = 0$, in which case $\min_x R ( x ) =
R ( t_0 )$.
\end{enumerate}
The next proposition asserts a bound on the weighted randomness of a
device, in terms of a rate curve $\pi$.  In order to facilitate the use
of the aforementioned principle, we will choose a weighting term
in the form $\pi' ( x )$ (where $\pi'$ denotes the derivative of $\pi$).

\begin{proposition}
\label{devindboundprop}
Let $G$ be a  game, and let $\pi$ be a rate curve for $G$.  Then, for all compatible
devices $D$ and all $r \in ( 0, W_G )$,
\begin{eqnarray}
R_G^{\epsilon, \pi' ( r ) } ( D ) & \geq & \pi ( r ) - \pi' (r ) r - O ( \epsilon).
\end{eqnarray}
\end{proposition}

\begin{proof}
We have the following:
\begin{eqnarray}
R_G^{\epsilon, \pi' ( r ) } ( D ) & \geq  &
R_G^{\epsilon} ( D ) - \pi' ( r ) W_G^{\epsilon} ( D ) - O ( \epsilon ) \\
& \geq & \pi ( W_G^{\epsilon}(D ) ) - \pi' ( r ) W_G^{\epsilon} ( D ) - O ( \epsilon ) \\
& \geq & \min_{t \in [0, W_G]} \left[ \pi ( t  ) - \pi' ( r ) t \right] - O ( \epsilon ) \\
& = & \pi (r ) - \pi' ( r ) r - O ( \epsilon ),
\end{eqnarray}
as desired.
\end{proof}

\subsection{The modified game $G_q$}

To analyze Protocol $R_{gen}$, it is helpful to define
a new unbounded game $G_q$, with input alphabet $\mathcal{I} := \{ 0, 1 \} \times \mathcal{A}$ and output alphabet $\mathcal{X}$,
which represents the procedure used
in Protocol $\overline{R}$.  Let 
\begin{eqnarray}
p_q ( (t, a) ) & = & \left\{ \begin{array}{cl} q \cdot p ( a, x ) &
\textnormal{ if } t = 1, \\  \\
(1-q) & \textnormal{ if } t = 0 \textnormal{ and } a = \overline{a}, \\ \\
0 & \textnormal{ if } t = 1 \textnormal{ and } a \neq \overline{a}.
\end{array}
\right. \\ \nonumber \\
H_q ( (t, a), x ) & = & \left\{ \begin{array}{cl} H( a, x )/q &
\textnormal{ if } t = 1, \\ \\
0 & \textnormal{ if } t = 0.
\end{array}
\right.  \label{newh}
\end{eqnarray}
(The inclusion of the denominator $q$ in (\ref{newh}) can be thought of
as compensation for the fact that game rounds only occur with frequency $1/q$.)
Also let $\rho_{\mathbf{i}}^{\mathbf{x}} = \rho_{\mathbf{a}}^{\mathbf{x}}$,
where $\mathbf{i} = ( (t_1, a_1), \ldots, (t_n , a_n ) )$, and define
$W^\epsilon_{G_q} ( D )$, $R^\epsilon_{G_q} ( D )$ and $R^{\epsilon,s}_{G_q} ( D )$ via these states.  We assert the following, which relates the difference
between the weighted and unweighted randomness of $G_q$ to the score
of the original game $G$.

\begin{proposition}
\label{qprop1}
Let $G$ be a  game and let
$s \in \mathbb{R}$.  Then, for all compatible devices $D$, all
$q \in (0, 1)$, and all $\epsilon \in (0, 1 ]$,
\begin{eqnarray}
R_{G_q}^{\epsilon,s} ( D ) \geq 
R_{G_q}^{\epsilon} ( D ) - s W_{G}^{\epsilon} ( D ) - O ( \epsilon /q).
\end{eqnarray}
\end{proposition}

\begin{proof}
The proof is similar to that of Proposition~\ref{devrelnprop}, but with a little
more care taken with the error term.  We must take into
 account the fact
that the scores in the game $H_q ( a, x )$ grow as $\Theta ( 1/q)$ as $q \to 0$.  

Again, it suffices to prove the result for abstact devices $D$ satisfying $\left< \rho \right>_{1+\epsilon} = 1$.  Then,
\begin{eqnarray}
&& \sum_{i,x} p_q ( i ) 2^{\epsilon s H_q ( i, x )} \left< \rho_i^x \right>_{1+\epsilon}  \\
& = &
\sum_{i,x} p_q ( i )  \left< \rho_i^x \right>_{1+\epsilon} + 
\sum_{i,x} p_q ( i )  (2^{\epsilon s H_q ( i, x )} - 1 ) \left< \rho_i^x \right>_{1+\epsilon} \\
\label{appline}
& = &
\sum_{i,x} p_q ( i)  \left< \rho_a^x \right>_{1+\epsilon} + 
q \sum_{a,x} p ( a )  (2^{\epsilon s H (a, x )/q} - 1 ) \left< \rho_a^x \right>_{1+\epsilon}.  \\
& \leq &
\sum_{i,x} p_q ( a )  \left< \rho_i^x \right>_{1+\epsilon} + 
q [ \epsilon s (\ln 2 ) W^\epsilon_G  ( D)/q + O ( (\epsilon/q)^2 ) ] \\
& \leq & 1 - \epsilon (\ln 2 ) R_{G_q}^\epsilon  ( D ) + \epsilon s (\ln 2 )
 W^\epsilon_G ( D ) + O ( \epsilon^2/q ).
\end{eqnarray}
Applying the function $-\frac{1}{\epsilon} \log ( \cdot )$ to both sides of the bound above yields the result.
\end{proof}

\begin{proposition}
\label{qprop2}
Let $G$ be a  game.  Then, for all compatible devices $D$, all
$q \in (0, 1)$, and all $\epsilon \in (0, 1 ]$,
\begin{eqnarray}
R_{G_q}^\epsilon ( D ) & \geq & R_{\overline{a}}^\epsilon ( D ) - O ( q ).
\end{eqnarray}
\end{proposition}

\begin{proof}
Observe the following:
\begin{eqnarray}
\left( \frac{ \sum_i \left< \rho_i^x
\right>_{1+\epsilon} }{\left< \rho \right>_{1+\epsilon} } \right) & = & (1-q) \left( \frac{ \sum_x \left< \rho_{\overline{a}}^x
\right>_{1+\epsilon} }{\left< \rho \right>_{1+\epsilon} } \right)
+ q \left( \frac{ \sum_{a,x} p ( a ) \left< \rho_a^x
\right>_{1+\epsilon} }{\left< \rho \right>_{1+\epsilon} } \right) \\
& \leq & (1-q) \left( \frac{ \sum_x \left< \rho_{\overline{a}}^x
\right>_{1+\epsilon} }{\left< \rho \right>_{1+\epsilon} } \right)
+ q \\
& \leq & ( 1 + O ( q \epsilon ) ) \left( \frac{ \sum_x \left< \rho_{\overline{a}}^x
\right>_{1+\epsilon} }{\left< \rho \right>_{1+\epsilon} } \right).
\end{eqnarray}
Applying the function $-\frac{1}{\epsilon} \log ( \cdot )$ to both sides yields
the desired result.
\end{proof}

Combining Propositions~\ref{qprop1} and \ref{qprop2}, we have the following.
\begin{proposition}
Let $G$ be a  game, and let
$s \in \mathbb{R}$.  Then, for all compatible devices $D$, all
$q \in (0, 1)$, and all $\epsilon \in (0, 1 ]$,
\begin{eqnarray}
R^{\epsilon, s}_{G_q} ( D ) & \geq & R^\epsilon_{\overline{a}} ( D ) - s W_G^\epsilon  ( D ) - O ( q + \epsilon/q ). \qed
\end{eqnarray}
\end{proposition}

The next proposition now follows by the same reasoning that was
used to prove Proposition~\ref{devindboundprop}.

\begin{proposition}
\label{weightedkeyprop}
Let $G$ be a  game, and let
$\pi$ be a rate curve for $(G, \overline{a})$.  Then,
for all compatible devices $D$, and all $r \in ( 0, W_G )$, $q \in (0, 1)$
and $\epsilon \in (0, 1 ]$,
\begin{eqnarray}
R_{G_q}^{\epsilon, \pi' ( r ) } ( D ) & \geq & \pi ( r ) - \pi' (r ) r - O ( q + \epsilon/q).  \qed
\end{eqnarray}
\end{proposition}

\subsection{The randomness of the success state of Protocol $R_{gen}$}

Let $E$ be a quantum system which purifies the device $D$ at the beginning
of Protocol $R_{gen}$.  
We use the following notation: let $X$ denote a classical register
containing the output sequence $(x_1, \ldots, x_N)$
at the conclusion of Protocol $R_{gen}$, let $I$ denote a classical
register containing the sequence \[ ((t_1, a_1  ) , \ldots , ( 
t_N , a_N )), \] and let $\Gamma_{IXE}$ denote the
joint state of $IXE$ at the conclusion of Protocol $R_{gen}$.
Let $s \subseteq \mathcal{I}^N \times \mathcal{X}^N$ denote the event that Protocol $R_{gen}$ succeeds.  We can
write the corresponding state as
\begin{eqnarray}
\Gamma_{IXE}^s & = & \sum_{( \mathbf{i}, \mathbf{x} ) 
\in s }  \left|
\mathbf{ix } \right> \left< \mathbf{ix} \right| 
\otimes \rho_{\mathbf{i}}^{\mathbf{x}}.
\end{eqnarray}

The next theorem addresses the $(1+\epsilon)$-randomness of the success state
of Protocol $R_{gen}$.
\begin{theorem}
\label{succstatethm}
Fix a  game $G$ and a rate
curve $\pi$ for $(G, \overline{a})$.
Then, 
\begin{eqnarray}
-\frac{1}{\epsilon}
\log \left( \frac{ \sum_{( \mathbf{i}, \mathbf{x} ) \in s } p_q ( \mathbf{i} ) \left< \rho_{\mathbf{i}}^{\mathbf{x}} \right>_{1+\epsilon}
}{\left< \rho \right>_{1+\epsilon}} \right) \geq
N \left[  \pi ( \chi ) - O ( q + \epsilon / q ) \right].
\end{eqnarray}
\end{theorem}

\begin{proof}
Applying Proposition~\ref{weightedkeyprop} with induction on $N$, we find that
\begin{eqnarray}
-\frac{1}{\epsilon} \log \left( \frac{ \sum_{\mathbf{i}, \mathbf{x}} p_q ( \mathbf{i} )
2^{ \epsilon \pi' ( \chi ) H_q ( \mathbf{a} , \mathbf{ x} ) }
\left< \rho_{\mathbf{i}}^{\mathbf{x}} \right>_{1+\epsilon} }{\left< \rho \right>_{1+\epsilon}} \right) \geq N [ \pi ( \chi ) - \pi' ( \chi ) \chi - O ( q + \epsilon / q ) ].
\end{eqnarray}
The event $s$ is determined by the inequality $H_q ( \mathbf{i} , \mathbf{x} ) \geq \chi N$.
The inequality above is thus preserved when we drop all terms corresponding to 
$(\mathbf{i} , \mathbf{x} ) \notin s$ from the summation, and in those that remain,
replace
the subterm $H_q ( \mathbf{i} , \mathbf{x} )$ with $\chi N$:
\begin{eqnarray}
-\frac{1}{\epsilon} \log \left( \frac{ \sum_{(\mathbf{i}, \mathbf{x})
\in s} p_q ( \mathbf{i} )
2^{ \epsilon \pi' ( \chi ) \chi N }
\left< \rho_{\mathbf{i}}^{\mathbf{x}} \right>_{1+\epsilon} }{\left< \rho \right>_{1+\epsilon}} \right) \geq N [ \pi ( \chi ) - \pi' ( \chi ) \chi - O ( q + \epsilon / q ) ].
\end{eqnarray}
Adding $N \pi' ( \chi ) \chi$ to both sides yields the desired result.
\end{proof}

\subsection{Extractable bits}

\label{extractablesubsec}

Our goal is now to use the content of subsection~\ref{smminsubsec} to 
prove a lower bound on the extractable bits in the output of Protocol $R_{gen}$.
First we will obtain a randomness lower bound for adversary states in the form shown on the right side of
(\ref{expabove}).
This is done by applying Theorem~\ref{succstatethm} to a class of modified devices.

\begin{proposition}
\label{succstateprop}
Fix a  game $G$ and a rate
curve $\pi$ for $(G, \overline{a})$.
Then, 
\begin{eqnarray}
\label{therelbound}
-\frac{1}{\epsilon}
\log \left(  \sum_{( \mathbf{i}, \mathbf{x} ) \in s } p_q ( \mathbf{i} ) \left< 
\rho^{\frac{-\epsilon}{2 + 2 \epsilon}} 
\rho_{\mathbf{i}}^{\mathbf{x}} \rho^{\frac{-\epsilon}{2 + 2 \epsilon}} \right>_{1+\epsilon}
 \right) \geq
N \left[  \pi ( \chi ) - O ( q + \epsilon / q ) \right].
\end{eqnarray}
\end{proposition}

\begin{proof}
Let $D_\epsilon$ be the abstract device that arises from $D$ by replacing the initial state
$\phi$ with the operator $\phi^{\frac{1}{1+\epsilon}}$.  Then, (using the notation 
of subsection~\ref{devmodelsubsec})
 the adversary state of
$D_\epsilon$ for any $\mathbf{i}, \mathbf{x}$ is isomorphic to
\begin{eqnarray}
\label{subnormexp}
\left( \phi^{\frac{1}{2+2\epsilon}} M_1^* M_2^* \cdots M_n^*
M_n \cdots M_2 M_1 \phi^{\frac{1}{2+2\epsilon}} \right)^\top & = & 
\rho^{\frac{-\epsilon}{2+2\epsilon}} \rho_{\mathbf{i}}^{\mathbf{x}} 
\rho^{\frac{-\epsilon}{2+2\epsilon}},
\end{eqnarray}
where $M_j = U_{a_j} P_{a_j}^{x_j}$.  Applying Theorem~\ref{succstatethm}
to $D_\epsilon$ (using the observation that $\left< \rho^{1/(1+\epsilon)} \right>_{1+\epsilon}
= \Tr ( \rho ) = 1$) yields the desired result.
\end{proof}

We note that the quantity on the lefthand side of (\ref{therelbound}) can be
expressed in the form of (\ref{expabove}), with $\Lambda = \Gamma^s_{XIE}$
and
\begin{eqnarray}
\sigma = \left( \sum_{\mathbf{i}} p_q ( \mathbf{i} ) \left| \mathbf{i}
\right> \left< \mathbf{i} \right| \right) \otimes \rho.
\end{eqnarray}
We are now ready to prove our main security results.

\begin{theorem}
\label{acentralsecthm}
Fix a  game $G$ and a rate curve
$\pi$ for $(G, \overline{a})$.  Then, for any $\delta > 0$,
\begin{eqnarray}
\frac{H^{\delta}_{min} \left( X \mid IE \right)_{\Gamma^s}}{N} \geq
\pi ( \chi ) - O \left(  q + \sqrt{ \frac{\log ( 2/\delta^2 ) }{qN}} \right).
\end{eqnarray}
\end{theorem}

\begin{proof}
Combining Theorem~\ref{translatetominthm} and Proposition~\ref{succstateprop}
yields
\begin{eqnarray}
\label{stepforward}
\frac{H^{\delta}_{min} \left( X \mid IE \right)_{\Gamma^s}}{N} & \geq &
 \pi ( \chi ) - O \left(  q + \epsilon/q + \frac{ \log ( 2 / \delta^2 ) }{N \epsilon} \right).
\end{eqnarray}
for all $\epsilon \in (0, 1 ]$.  We can minimize the big-$O$ expression in 
(\ref{stepforward}) by setting $\epsilon$ so that the second
and third terms become equal.  Let
\begin{eqnarray}
\epsilon = \min \left\{ 1, \sqrt{\frac{ q \log ( 2 / \delta^2 ) }{N}} \right\},
\end{eqnarray}
and the desired result holds.
\end{proof}

\begin{corollary}
Fix a  game $G$ and a rate curve
$\pi$ for $(G, \overline{a})$.  Then, for any $b > 0$,
\begin{eqnarray}
\label{finalcorexp}
\frac{H^{\sqrt{2} \cdot 2^{-bqN}}_{min} \left( X \mid IE \right)_{\Gamma^s}}{N} \geq
\pi ( \chi ) - o ( 1 ),
\end{eqnarray}
where $o(1)$ denotes a function of $(q, b)$ that tends to zero as $(q, b ) \to (0, 0)$.
\end{corollary}

\begin{proof}
This is immediate from Theorem~\ref{acentralsecthm} by substitution.
\end{proof}

For the final security statement, we use the terminology of extractable
bits discussed in subsection~\ref{centralresultsubsec}.

\begin{corollary}
\label{finalcor}
Fix a  game $G$ and a rate curve
$\pi$ for $(G, \overline{a})$.
Then for any $b > 0$, Protocol $R_{gen}$ produces at least
\begin{eqnarray}
N [ \pi ( \chi ) - o ( 1 ) ]
\end{eqnarray}
extractable bits with soundness error $3 \cdot 2^{-bqN}$, where $o(1)$ denotes
a function of $(q,b)$ (which can depend on $G, \pi$) that vanishes
as $(q, b ) \to ( 0, 0 )$.
\end{corollary}

\begin{proof}
Find a classical-quantum operator $\Gamma' \geq 0$ with $\left\| \Gamma' - \Gamma^s \right\|_1
\leq \sqrt{2} \cdot 2^{-bqN}$ whose min-entropy
$H_{min} ( X \mid IE)_{\Gamma'}$
is equal to the 
numerator in \ref{finalcorexp}.  If $\Tr (  \Gamma^s ) \leq 3 \cdot 2^{- bqN}$,
then the assertion of the corollary is trivial.  If not, then
$\Tr ( \Gamma' ) > 2^{-bqN}$, and
\begin{eqnarray}
\frac{ H_{min} \left( X \mid IE \right)_{\Gamma' / \Tr ( \Gamma') }}{N} & \geq &
\pi ( \chi ) - o ( 1 ) + \frac{\log \Tr ( \Gamma' )}{N} , \\
& \geq & \pi ( \chi ) - o ( 1 ),
\end{eqnarray}
as desired.
\end{proof}

\section{Acknowledgements}

Many thanks to Dong-Ling Deng and Kihwan Kim for sharing with us their
work on randomness expansion, and
Patrick Ion for introducing us to the literature on Kochen-Specker
inequalities.  We also thank Karl Winsor, Xiao Yuan, Qi Zhao, Zhu Cao and
Christopher Portmann for discussions that helped improve technical 
aspects of the draft.

\newpage

\appendix

\section{Not all superclassical distributions are randomness generating}

\label{magicapp}

The Magic Square game $M$ is a two player game defined as follows:
\begin{enumerate}
\item The input and output alphabets are $\{ 0, 1, 2 \}$ and $\{ 0, 1 \}^3$,
respectively for each player.

\item The input probability distribution is uniform.

\item The game is won iff the inputs $a_1$ and $a_2$ and outputs
$(x_1^{(0)}, x_1^{(1)}, x_1^{(2)} )$ and $(x_2^{(0)}, x_2^{(1)}, x_2^{(2)} )$
satisfy all of the following
\begin{eqnarray}
\label{mag1} x_1^{(0)} \oplus x_1^{(1)} \oplus x_1^{(2)} & = & 0 \\
\label{mag2} x_2^{(0)} \oplus x_2^{(1)} \oplus x_2^{(2)} & = & 1 \\
\label{mag3} x_1^{(a_1)} & = & x_2^{(a_2)}.
\end{eqnarray}
\end{enumerate}

The reason for the name ``magic square'' is this: the game is for
the first player to fill in one row of bits in a $3 \times 3$ square while
the second player is asked to fill in one column.  The requirements
are that the row has even parity, the column has odd parity, and the
overlapping bits agree.  The maximal classical score for this game is
$8/9 \approx 0.888$.

We construct a device for this game.  Let $Q = Q_1 \otimes Q_2
= \mathbb{C}^2 \otimes \mathbb{C}^2$, let $\alpha$ be the projector
of $Q$ onto the Bell state $\frac{1}{\sqrt{2}} \left( \left| 00 \right> + \left|
11 \right> \right)$.  Let the measurement strategy $\{ \{ P_{1, a}^{\mathbf{x}} \}_\mathbf{x} \}_a$ for the first player be given
by
\begin{eqnarray}
P_{1, 0}^{000} & = & \mathbb{I} \\
P_{1, 1}^{000} & = & \left| 0 \right> \left< 0 \right| \\
P_{1, 1}^{011} & = & \left| \pi/2 \right> \left< \pi/2 \right| \\
P_{1, 2}^{101} & = & \left|  -\pi/4 \right> \left< -\pi/4 \right|, \\
P_{1, 2}^{110} & = & \left|  \pi/4 \right> \left< \pi/4 \right|,
\end{eqnarray}
where we have used the notation $\left| \theta \right> =
(\cos \theta ) \left| 0 \right> + ( \sin \theta) \left| 1 \right>$.  Let
the measurement strategy for the second player be given by
\begin{eqnarray}
P_{2, 0}^{001} & = & \mathbb{I} \\
P_{2, 1}^{001} & = & \left| \pi/8 \right> \left< \pi/8 \right| \\
P_{2, 1}^{010} & = & \left| 5 \pi/8 \right> \left< 5 \pi/8 \right| \\
P_{2, 2}^{001} & = & \left|  -\pi/8 \right> \left< -\pi/8 \right|, \\
P_{2, 2}^{010} & = & \left|  3 \pi/8 \right> \left< 3\pi/8 \right|.
\end{eqnarray}
The winning probability is $1$ when the
inputs are such that $(a_1 = 0) \vee (a_2 = 0)$,
and the winning probability is $\frac{1}{2} + \frac{\sqrt{2}}{4}$ otherwise.
(In the latter case the players are essentially conducting
an optimal strategy for the CHSH game.)  Thus the average 
winning probability across all inputs is
\begin{eqnarray}
(5/9) + (4/9)\left( \frac{1}{2} + \frac{\sqrt{2}}{4} \right)
& \approx & 0.934.
\end{eqnarray}

Similarly, for any output pair $(\overline{x}_1, \overline{x}_2)$ satisfying the conditions
\begin{eqnarray}
\label{magalt1} \overline{x}_1^{(0)} \oplus \overline{x}_1^{(1)} \oplus \overline{x}_1^{(2)} & = & 0 \\
\label{magalt2} \overline{x}_2^{(0)} \oplus \overline{x}_2^{(1)} \oplus \overline{x}_2^{(2)} & = & 1 \\
\label{magalt3} \overline{x}_1^{(0)} & = & \overline{x}_2^{(0)}.
\end{eqnarray}
construct an analogous device, which we denote $E^{\overline{x}_1,\overline{x}_2}$,
which always generates output $(\overline{x}_1, \overline{x}_2)$ on input $(0,0)$
and wins with overall probability $\approx 0.934$.  
Let $E$ be a quantum device which chooses an output pair $(\overline{x}_1, \overline{x}_2 )$
according to a uniform distribution
among pairs satisfying (\ref{magalt1}--\ref{magalt3}) and then behaves
identically to $E^{\overline{x}_1,\overline{x}_2}$.  (Note that
if an adversary possesses knowledge of
the choice of $(\overline{x}_1, \overline{x}_2)$, the outputs of such a device
on input $(0,0)$ will be fully predictable to her.) 

For each pair $(\overline{a}_1, \overline{a}_2) \in \{ 1, 2, 3 \}^2$, it is possible to
construct a completely classical two-part classical device
$D^{\overline{a}_1 \overline{a}_2}$ which performs as follows: if the input
pair is anything other than $(\overline{a}_1, \overline{a}_2)$,
the device generates a uniform distribution over all pairs of sequences
satisfying (\ref{mag1}--\ref{mag3}), and if the input pair
is $(\overline{a}_1, \overline{a}_2)$, the device generates a uniform distribution over all
pairs of sequences that satisfy (\ref{mag1}--\ref{mag2}) but
do not satisfy (\ref{mag3}).  Let $S = \{ (\overline{a}_1, \overline{a}_2) \mid (\overline{a}_1
= 0) \vee
(\overline{a}_2 = 0 ) \}$, and let $D^S$ denote a classical device
which, on any input,
chooses an element $(\overline{a}_1, \overline{a}_2) \in S$ uniformly at random
and behaves as $D^{(\overline{a}_1, \overline{a}_2)}$.  

The device $D^S$ loses the Magic Square game with probability
$1/5 = 0.2$ if its given input pair is in $S$, and loses with probability
$0$ if its given input pair is not in $S$.  The device $E$ constructed
above loses with probability $0$ if the given input pair is in $S$,
and loses with probability $\beta := (\frac{1}{2} - \frac{\sqrt{2}}{4} )
\approx 0.146$ otherwise.
Let $D$ be a two-part quantum device which behaves identically to
$D^S$ with probability $\beta/(0.2 + \beta )$, and  behaves identically to
$E$ with probability $0.2/(0.2 + \beta )$.  Then, the losing probability
for $D$ on any input pair is
\begin{eqnarray}
\frac{ 0.2 \beta }{0.2 + \beta} & \approx & 0.0845.
\end{eqnarray}

By construction, the output of $D$ on input $(0,0)$ is fully predictable
to an adversary who possess appropriate classical information.
On the other hand, by the symmetry of the conditional
input-output distribution of $D$, it is possible to take any input
pair $(\overline{a}_1, \overline{a}_2)$ and construct a simulation of $D$ for which
the output of $(\overline{a}_1, \overline{a}_2)$ is fully predictable.  Therefore,
even though $D$ achieves a superclassical score
at the Magic Square game, it is not randomness generating.

(See \cite{HorodeckiH:2010} for a related calculation involving
the magic square game.)

\section{The noise tolerance for the CHSH game}

Let $CHSH$ denote the two-player game in which all alphabets
$\mathcal{A}_i$ and $\mathcal{X}_i$ are equal to $\{ 0, 1 \}$,
the input distribution is uniform, and
\begin{eqnarray}
H ( \mathbf{a} , \mathbf{x} ) & = & \neg \left( x_1 \oplus x_2 \oplus (a_1 \wedge a_2) \right).
\end{eqnarray}

\begin{proposition}
\label{chshprop}
The quantity $W_{CHSH, 00}$ is equal to $3/4$.
\end{proposition}

\begin{proof}
Let $D$ be a device compatible with the CHSH game
which gives deterministic outputs on input $00$.
We wish to prove that the losing probability of $D$ is at least $1/4$.
Let $A_1, A_2, X_1, X_2$ denote random variables representing
the inputs and outpus of $D$.  

Note that for any binary random variables $X,Y$, the probability
of the event $X \neq Y$ is at least $\left| \mathbf{P} ( X = 0 ) - \mathbf{P} ( Y = 0 ) \right|$
and the probability of the event $X = Y$ is at least $\left| \mathbf{ P} ( X = 0 ) - 1 + \mathbf{P} ( Y = 0) \right|$.  Observe the following inequalities, where we use the shorthand
$p_u^v = \mathbf{P} ( X_1 = v \mid A_1 = u )$ and
$q_u^v = \mathbf{P} ( X_2 = v \mid A_2 = u )$:
\begin{eqnarray*}
\mathbf{P}_D ( X_1 \oplus X_2 \neq A_1 \wedge A_2 ) & = &
\frac{1}{4} \left[ \mathbf{P}_D ( X_1 \neq X_2 \mid A = (0, 0 ) ) + 
\mathbf{P}_D ( X_1 \neq X_2 \mid A = (0, 1 ) ) \right. \\
&& \left. +  
\mathbf{P}_D ( X_1 \neq X_2 \mid A = (1, 0 ) )  +
\mathbf{P}_D ( X_1 = X_2 \mid A = (1, 1 ) ) \right]\\
& \geq & \left| p_0^0 - q_0^0 \right| + \left| p_0^0 - q_1^0 \right|
+ \left| q_0^0 - p_1^0 \right| + \left| q_1^0 - 1 +  p_1^0 \right| \\
& \geq & \frac{1}{4} \left| 2 p_0^0 - 1 \right|
\end{eqnarray*}
where we used the triangle inequality in the last step.  The final quantity is
equal to $1/4$, and this completes the proof.
\end{proof}

Thus $W_{CHSH, 00} = 0.75$, while $W_{CHSH} = \frac{1}{2} +
\frac{\sqrt{2}}{4} \approx 0.853$.  Therefore the noise
tolerance
provided by Theorem~\ref{mainthmsimple} for the CHSH game is approximately
$10.3\%$.

\newpage
\bibliographystyle{abbrv}
\bibliography{../quantumsec}

\end{document}